\documentclass{birkjour}
\usepackage[T1]{fontenc}
\usepackage[utf8]{inputenc}
\usepackage{amsfonts,amssymb,amsthm,amscd}
\usepackage{latexsym,graphicx,cite,enumitem}
\renewcommand{\[}{\begin{equation}}
\renewcommand{\]}{\end{equation}}

\newtheorem*{thm}{Theorem}
\newtheorem*{lemma}{Lemma}

\theoremstyle{definition}
\newtheorem*{Def}{Definition}

\renewcommand{\Psi}{\sigma}

\newcommand{\D}{{D\mkern-11.5mu/\,}} \newcommand{\Q}{\mathbb{H}}
\newcommand{\R}{\mathbb{R}}
\newcommand{\CC}{\mathbb{C}}

\newcommand{\zero}{\phantom{}}
\newcommand{\Z}{\mathbb{Z}}

\newcommand{\id}{\mathsf{id}}

\newcommand{\Cl}{\mathcal{C}\ell}

\begin{document}

\title[Fermion multiplet as internal noncommutative forms]
{On noncommutative geometry\\
  of the Standard Model:\\
  fermion multiplet as internal forms}

\author[L. Dąbrowski]{Ludwik Dąbrowski}

\address{Scuola Internazionale Superiore di Studi Avanzati (SISSA)\\
  via Bonomea 265\\ I-34136 Trieste}

\email{dabrow@sissa.it}

\thanks{This work is part of the project Quantum Dynamics sponsored by
  EU-grant RISE 691246 and Polish Government grant 317281, and was
  partially supported by by the Polish Ministry of Science and Higher
  Education 2015-2019 matching fund 3542/H2020/2016/2.}



\date{November 11, 2017}
---------------------------------------------------------------------

\begin{abstract}
  We unveil the geometric nature of the multiplet of fundamental
  fermions in the Standard Model of fundamental particles as a
  noncommutative analogue of de\,Rham forms on the internal finite
  quantum space.
\end{abstract}

\maketitle

\section{Introduction} From the conceptual point of view the Standard
Model (S.M.) of fundamental particles and their interactions is a
particular model of $U(1)\times SU(2)\times SU(3)$ gauge fields
(bosons) minimally coupled to matter fields (fermions), plus a Higgs
field (boson). After the second quantization with gauge fixing,
spontaneous symmetry breaking mechanism, regularization and
perturbative renormalization it extremely well concords with the
experimental data. Even so (unreasonably) successful it however does
not explain (though somewhat constrains) the list of particles, in
particular the existence of 3 families, contains several parameters
and does not include the fourth known interaction: gravitation, with
its own fundamental symmetry: general relativity or
diffeomorphisms. There have been various attempts to settle some of
the above shortcomings: GUT based on a simple group $SU(5)$ or
$SO(10)$, modern variants of old Kaluza-Klein model with
`compactified' internal dimensions, and others more recent and
fashionable, that are still under extensive massive study.

Of our interest in this note is another distinct approach to the
S.M. in the framework of noncommutative geometry by A.\,Connes et.al.,
see e.g. \cite{CM08}, which is not so widely known among physicists.
It interprets the multiplet of fundamental fermions as a field on a
finite quantum space, on which the would be ``coordinates'', as well
as the algebra of ``functions'' fail to commute. In this note we focus
on its deeper geometric structure, aiming to shed more light on the
nature of this internal quantum space. We shall review the key results
of the two recent papers \cite{dd16,dds17}, explaining in more detail
the classical (commutative) motivation behind them. For that some well
known material in differential geometry will be presented from the
viewpoint of the so called spectral triples; with the only new
contribution in the last part of Subsection~\ref{hdrst} regarding
their KO-dimension.

\section{Introduction}\label{intro} The noncommutative formulation
$\nu$S.M. of the Standard Model takes its cue from its geometry which
in mathematical terminology corresponds to a connection (locally a
multiplet of vector fields) whose structure group is $U(1)\times
SU(2)\times SU(3)$ on (a multiplet) of spinors, together with a
doublet of scalar fields.  Although it does not renounce of {\em groups}, $\nu$S.M. is however based primarily on {\em
algebras}. Moreover to the 75 years-old Gelfand-Naimark
(anti)equivalence:
$$\boxed{
\text{topological\; spaces} \longleftrightarrow \text{commutative }
C^*\text{-algebras}}
$$
\vspace{1mm} and to the Serre-Swan equivalence:
$$
\boxed{ \text{vector bundles} \longleftrightarrow \text{modules}}
$$
it adjoins two other ingredients to encode such structures as {\em
smoothness}, {\em calculus} and (Riemannian) {\em metric} on a space
$M$. The first one is a Hilbert space $H$ that carries a unitary
representation of a (possibly noncommutative) $*$-algebra $A$, and so
obviously also of its norm completed $C^*$-algebra. The second one is
an analogue of the Dirac operator on $H$. Together with a $*$-algebra
$A$ they satisfy certain analytic conditions: $D$ is selfadjoint,
$[D,a]$ are bounded $\forall a\in A$ and $(D-z)^{-1}$ are compact for
$z\in\CC\setminus \R$, so that they form the so called {\em
spectral\,triple}\, (S.T.) \cite{Con96}
$$(A,H,D).$$
Such a S.T. is {\em even} if there is a $\Z_2$-grading $\chi$ of $H$,
$\chi^2=1, \chi^\dagger=\chi$, with which all $a\in A$ commute and $D$
anticommutes. Furthermore it is {\em real} if there is a
$\CC$-antilinear isometric operator $J$ on $H$, such that denoting $B'$ the
commutant of $B\subset {\mathcal B}(H)$,
\begin{equation}\label{eq:0order} JAJ^{-1}\subset A', \vspace{-1mm}
\end{equation} which is often termed {\em order 0 condition}.  We say
that a real S.T. satisfies the {\em order 1 condition} if
\begin{equation}\label{eq:1order} JAJ^{-1} \subset [D,A]'
\end{equation} and the {\em order 2 condition}\footnote{while the
order 1 condition means classically that $D$ is order 1 differential
operator, order 0 and 2 conditions don't have such interpretation} if
\begin{equation}\label{eq:2order} J[D,A]J^{-1} \subset [D,A]'.
\end{equation}

The $A$-bimodule spanned by $[D,A]$ is often called space of 1-forms,
and the algebra generated by $A$ and $[D,A]$ the space of all forms
for the Dirac calculus (with the exterior derivative given by the
commutator with $D$). Motivated by the classical examples (cf. the
next subsections) slightly abusing the terminology we will call, quite
as in \cite{LRV12}, {\em Cliford algebra} the complex $*$-algebra
$\Cl_D(A)$ generated by $A$ and $[D,A]$.

Note that for noncommutative $A$ a priori there is no right action of
$A$ on $H$, but given $J$ there is one:
$$ha:=Ja^*J^{-1} h,$$
that commutes with the left action due to the order 0 condition and so
$H$ becomes an $A$-$A$ bimodule.  Furthermore, if the order 1
condition holds $H$ becomes a $\Cl_D(A)$-$A$ bimodule, and if the 2nd
order condition holds $H$ becomes even a $\Cl_D(A)$-$\Cl_D(A)$
bimodule.

Connes formulated few other important properties of real spectral
tri\-ples. One of them requires that the following identities are
satisfied
\begin{equation}\label{j2} J^2 = \epsilon\, \id,
\end{equation}
\begin{equation}\label{jd} DJ = \epsilon' JD,
\end{equation} and in even case
\begin{equation}\label{jg} \chi J = \epsilon''J\chi,
\end{equation} where the three signs $\epsilon,\epsilon', \epsilon''
\in \{+,-\}$ specify the so called KO-dimension modulo 8: \vspace{1mm}
\begin{equation}
\begin{array}{|c|c|c|c|c|c|c|c|c|} \hline n &0&1&2&3&4&5&6&7 \\ \hline
\epsilon &+&+&-&-&-&-&+&+\\ \hline \epsilon^{\prime}
&+&-&+&+&+&-&+&+\\ \hline \epsilon^{\prime\prime} &+&&-&&+&&-& \\
\hline
\end{array}
\end{equation} (If dimension $n$ is even one can alternatively use
$\chi J$ as a new real structure, which changes the parameter
$\epsilon^{\prime}$ to $-\epsilon^{\prime}$, and $\epsilon$ to
$\epsilon \epsilon^{\prime\prime}$).

\subsection{Canonical spectral triple}\label{cst} A prototype example
is the {\em canonical} S.T.  on a closed oriented \underline{spin}
manifold $M$ of dimension $n$ equipped with a Riemannian metric $g$:
\begin{equation}\label{can} (C^\infty(M), L^2(S),\D ).
\end{equation} Here $C^\infty(M)$ is the algebra of smooth complex
functions on $M$, $S$ is the rank$_\CC = 2^{[n/2]}$ bundle of Dirac
spinors on $M$, whose sections carry a faithful irreducible
representation
\begin{equation}\label{gamma}
\gamma:\Gamma(\Cl(M))\overset{\approx}{\longrightarrow}
\operatorname{End}_{C^\infty(M)} \Gamma(S) \approx
\Gamma(S)\otimes_{C^\infty(M)} \Gamma(S)^*,
\end{equation} of the algebra of sections of the (simple part of)
complex Clifford bundle $\Cl(M)$, generated by $v\in TM$ with relation
$v^2 + g(v,v)=0$.  Furthermore $\D$ is the usual Dirac operator
\begin{equation} \D= \gamma\circ \tilde\nabla = \sum_j^n \gamma(e_j) \tilde\nabla_{e_j},
\end{equation} where $\tilde\nabla$ is a lift to $S$ of the
Levi-Civita connection on $M$ and $e_j$, $j=1,\dots ,n$, is a local
oriented orthonormal basis of $TM$.

One has
$$[\D,f]= \gamma(\mathrm{d}f)
\, ,\quad f\in C^\infty(M),$$ or, what is the same, the symbol of $\D$
is
$$
\sigma_\D(\xi) = -i\gamma(\xi), \quad \xi\in T^*M\,,
$$
where we have identified $TM\approx T^*M$, and so $\nabla f$ with
$df$, using the metric $g$.  Note that the operators of that form
together with functions generate an isomorphic copy of the Clifford
algebra $\Gamma(\Cl(M))$.

If $\operatorname{dim} M$ is even there is also a $\Z_2$-grading
$\chi_S$ of $L^2(S)$, with which all $a\in A$ commute and $\D$
anticommutes.  It should be mentioned that $\D$ is an elliptic
operator and its index, or more precisely the Fredholm index of
$\D|_{\Gamma(S)}:\Gamma(S)\to \Gamma(S)$, where $\Gamma(S^{\pm})$ are
$\pm 1$ eigenspaces of $\chi_S$, plays an important role in geometry
and applications to physics. It can be expressed in terms of the
characteristic class called $\hat{A}$-genus, a topological invariant
of~$M$.

Furthermore there is a real structure (known as {\em charge
conjugation}) $J_S$ on $L^2(S)$, that satisfies the order 0 condition
\eqref{eq:0order} and the order 1 condition \eqref{eq:1order}, but
\underline{not} the order 2 condition \eqref{eq:2order}.  Indeed,
$J_S$ and \eqref{can} obey a stronger version of \eqref{eq:0order} and
\eqref{eq:1order} which excludes \eqref{eq:2order}. Namely the norm
closure of $C^\infty(M)$, that is the algebra $C(M)$ of continuous
functions on $M$, is the maximal commutant in $B(L^2(S))$ of the norm
closure of $\Cl_{\D}(C^\infty(M))=C^\infty(M)[\D, C^\infty(M)]$, which
is just the algebra of continuous sections $\Gamma(\Cl(M))$ of the
(complexified) Clifford bundle $\Cl(M)$ on $M$ in the Dirac
representation.  We can thus say that ``the Dirac spinor fields
provide a Morita equivalence $C(M)-\Gamma(\Cl(M))$ bimodule''.

As a matter of fact for the latter property it suffices that $M$ is
spin$_c$, that can be defined by any of the following three equivalent
statements:
\begin{enumerate}[label=\roman*)]
\item there exist a principal $\mathit{Spin}_c(n)$-bundle, such
that the vector bundle associated with the representation
$\rho\times\!\Box$ is isomorphic to the tangent bundle $T(M)$;
\item $SO(n)$-bundle of oriented orthogonal frames lifts to
$\mathit{Spin}_c(n)$;
\item the second Stiefel-Whitney class $w_2(M)$ is a modulo 2
reduction of a class in $H^2(M,\Z)$.
\end{enumerate} 
Here $\mathit{Spin}_c(n)$ is the quotient group
of $\mathit{Spin}(n)\times U(1)$ by the subgroup\linebreak
$\Z_2^{\text{diag}}= \{(1,1),(-1,-1)\}$, $\rho:
\mathit{Spin}(n)\to SO(n)$ is the nontrivial double covering
and\, $\Box:U(1)\to U(1)$ is the square map.

Importantly however the property that an oriented Riemannian manifold
$M$ is spin$_c$ is actually tantamount \cite{Ply86} to
\begin{equation}\label{spinc} \boxed{\exists \text{ a Morita
equivalence } {\Cl(M)\!-\!C(M)}~ \text{ bimodule }~ \Sigma .}
\end{equation} Indeed when \eqref{spinc} holds then automatically
$\Sigma\approx\Gamma(S)$, where $S$ is the $\CC$-vector bundle of
Dirac spinors on $M$.

Therefore spin$_c$ manifolds lend itself to noncommutative
generalization via the algebraic property \eqref{spinc} by taking
advantage of the definition of Clifford algebra $\Cl_D(A)$.  Next, the
algebraic characterization of {\em spin} manifolds also admits a
noncommutative generalization as the condition \eqref{spinc} plus a
real structure (charge conjugation) $J$ that implements it.

We remark that the canonical S.T. fully encodes the geometric data on
$M$, that can be indeed reconstructed \cite{Con08} from a commutative
S.T. with certain few additional properties.  One of these properties
requires that KO-dimension, defined by \eqref{j2}, \eqref{jd},
\eqref{jg}, is equal for the operators $\D, \chi_S, J_S$ to the
dimension of $M$ modulo 8.

\subsection{Hodge-de\,Rham spectral triple}\label{hdrst} The canonical
S.T. \eqref{can} is not the only natural S.T.  On any oriented closed
Riemannian manifold $M$ there is also
\begin{equation}\label{eq:form} (C^\infty(M), L^2(\Omega(M)),d+d^* ),
\end{equation} where $\Omega(M)$ is the graded space of complex
de\,Rham differential forms on $M$, $d$ is the exterior differential
and $d^*$ is its adjoint.

The operator $d+d^*$ is actually Dirac-type since
\begin{equation}\label{la} d+d^*= \lambda\circ \nabla,
\end{equation} where
\begin{equation}\label{lambda} \lambda: \Gamma(\Cl(M))\to
\operatorname{End}_{C^\infty(M)} \Omega(M),\,\, \lambda(v)=v\wedge -
v\lrcorner\,,\,\, v\in T^*M\approx TM,
\end{equation} is the (reducible) faithful complex representation on
$\Omega(M)$ of the algebra of sections $\Gamma(\Cl(M))$ of the
(complexified) Clifford bundle $\Cl(M)$ over $M$.  The formula
\eqref{la} means that
\begin{equation}\label{lamb} [d+d^*,f]=\lambda(\mathrm{d}f) \,
,\quad f\in C^\infty(M),
\end{equation} what is also the same as the symbol of $d+d^*$ being
$$
\sigma_{d+d^*}(\xi) = -i\lambda(\xi), \quad \xi\in T^*M.
$$
Note that as for the canonical $\D$ the operators of the form
\eqref{la} together with functions generate $\Gamma(\Cl(M))$, and thus
indeed $\Cl_{d+d^*}(C^\infty(M))\approx \Gamma(\Cl(M))$.

The representation $\lambda$ is equivalent to the left regular
self-representation of $\Gamma(\Cl(M))$, via the isomorphism of vector
bundles $\Cl(M)\approx \Omega(M)$.  There is also an
anti-representation
\begin{equation}\label{rho} \rho:\Gamma(\Cl(M))\to
\operatorname{End}_{C^\infty(M)} \Omega(M),\,\, \lambda(v)=(v\wedge +
v\lrcorner)\,\chi_\Omega\,,\,\, v\in T^*M\approx TM,
\end{equation} where
\begin{equation}\label{chi} \chi_\Omega = \pm 1
\end{equation} on even forms $\Omega(M)^{\text{even}}$, respectively
odd forms $\Omega(M)^{\text{odd}}$.  It is equivalent to the right
regular self-antirepresentation of $\Gamma(\Cl(M))$.

Furthermore, since the endomorphisms $\lambda(v)$ and $\rho(v')$ commute, $\Omega(M)$ is a
$\Gamma(\Cl(M))$-$\Gamma(\Cl(M))$ bimodule, which is
equivalent to $\Gamma(\Cl(M))$.  Thus, $\Omega(M)$ is actually a
self-Morita equivalence $\Gamma(\Cl(M))$-$\Gamma(\Cl(M))$ bimodule,
the property which in fact provides its unambiguous characterization up to
a tensor product with sections of a complex line bundle.

The operator $\chi_\Omega$ \eqref{chi} always defines a $\Z_2$-grading
of $L^2(\Omega(M))$ according to the parity of forms. If dim$M=n=2m$
(even) there is also another grading $\chi_\Omega'$ given by the
normalized Hodge star operator, defined in terms of a local
orthonormal oriented basis $e^j$, $j=1, \dots, n$, $n=2m$ of $T^*M$ by
\begin{equation} \chi_\Omega' (e^{j_1}\wedge\dots \wedge e^{j_k}) =
i^{k(k-1)+m}e^{j_{k+1}}\wedge\dots \wedge e^{j_n}, \quad 0\leq k\leq
n,
\end{equation} where $j_1, \dots j_n$ is an even permutation of $1,
\dots n$.

Both the gradings $\chi_\Omega$ and $\chi_\Omega'$ commute with $a\in
C^\infty(M)$ and anticommute with $d+d^*$.  As well known, they play
important role for the index of the elliptic operator $d+d^*$.  More
precisely the Fredholm index of
$$(d+d^*)|_{\Omega(M)^{\text{even}}}:\Omega(M)^{\text{even}}\to \Omega(M)^{\text{odd}}$$
computes the Euler character of $M$, while the index of
$$(d+d^*)|_{\Omega(M)^{s}}: \Omega(M)^{s}\to \Omega(M)^{a},$$
where $\Omega(M)^{s}$ and $\Omega(M)^{a}$ are respectively the $\pm 1$
eigenspaces of $\chi_\Omega'$, computes the signature of $M$.

Furthermore there is also a real structure $J_\Omega$ on
$L^2(\Omega(M))$ given just by the complex conjugation of forms.  It
satisfies the conditions \eqref{eq:0order} and \eqref{eq:1order} but
definitely \underline{not} \eqref{eq:2order} and therefore can not
implement the $\Gamma(\Cl(M))$-$\Gamma(\Cl(M))$ self-Morita
equivalence as above.

In order to implement this equivalence, we need another real structure
$J_\Omega'$ on $\Omega(M)$, that interchanges the actions $\lambda$
and $\rho$. It turns out that it can be defined as
\begin{equation} J_\Omega'(e^{j_1}\wedge\dots \wedge e^{j_k}) =
e^{j_{k}}\wedge\dots \wedge e{j_1}, \quad 0\leq k\leq n,
\end{equation} which corresponds to the main anti-involution on
$\Gamma(\Cl(M))$ and can be simply written on $\Omega^k(M)$ as
\begin{equation} J_\Omega' = (-)^{k(k-1)/2}\circ c.c.
\end{equation} This real structure satisfies all the order conditions
\eqref{eq:0order}, \eqref{eq:1order} and \eqref{eq:2order} and does
implement the $\Gamma(\Cl(M))$-$\Gamma(\Cl(M))$ self-Morita
equivalence as above.

We mention that for the operators $d+d^*, \chi_\Omega, J_\Omega$ one
gets the signs $\epsilon=1, \epsilon'=1, \epsilon''=1 $ and so the
KO-dimension is 0.  Instead for the operators $d+d^*, \chi_\Omega',
J_\Omega$ the signs are $\epsilon=1, \epsilon'=1, \epsilon''=(-1)^m $
and so the KO-dimension is 0 if n=0 mod 4, and 6 if n=2 mod 4
[A. Rubin, MSc Thesis].

As far as the operators $d+d^*, \chi_\Omega, J_\Omega'$ are concerned
we get the signs $\epsilon=1, \epsilon'=1, \epsilon''=1 $ and so the
KO-dimension is 0.  Instead for the operators $d+d^*, \chi_\Omega',
J_\Omega'$ we get the signs $\epsilon=1$ and $\epsilon'=1$, while on
$\Omega^k(M)$ we obtain
\begin{equation} J_\Omega'\chi_\Omega' = (-1)^k \chi_\Omega'J_\Omega'
\end{equation} in which the sign appears that depends on $k$. Thus in
that case there is no overall sign $\epsilon''$ but in fact a grading
given by $\chi_\Omega$. This feature generalizes somewhat the notion
of KO-dimension which as known was tailored for the canonical spectral
triple.

Closing this section we remark that it is not clear whether, and with
which additional conditions, this S.T. equiped with any combination of
the gradings and real structures as above may faithfully encode the
geometric data on $M$, that can be then reconstructed.

\section{Noncommutative formulation of the Standard Model:
$\nu$S.M.}\label{nsm}

Concerning the underlying arena of $\nu$S.M., see e.g. \cite{CM08}, it
is
\begin{center} (ordinary (spin) manifold $M$) $\times$ (finite quantum
space $F$),\\
\end{center} described by the algebra $C^\infty(M)\otimes A_F\approx
C^\infty(M, A_F)$, where
$$\boxed{
A_F=\CC\oplus\Q\oplus M_3(\CC).}
$$
Here $\Q$ is the (real) algebra of matrices of the form
$$
\begin{bmatrix} \alpha & \beta \\ -\bar\beta &
\;\bar\alpha \end{bmatrix} \;,\qquad\alpha,\beta\in\CC,
$$
which is isomorphic to the algebra of quaternions.

The Hilbert space is
$$L^2(S)
\otimes\, H_F,$$ where
$$\boxed{
H_F= \CC^{96} =: H_f\otimes \CC^{3},}
$$
with $\CC^{3}$ corresponding to 3 generations, and
$$H_f= \CC^{32}\simeq M_{8\times 4}(\CC).$$
The orthonormal basis of $H_f$ is labelled by particles and
antiparticles, that we arrange as a $8\times 4$ matrix
$$
\begin{bmatrix} \nu_R & u^1_R & u^2_R & u^3_R \\ e_R & d^1_R & d^2_R &
d^3_R \\ \nu_L & u^1_L & u^2_L & u^3_L \\ e_L & d^1_L & d^2_L &
d^3_L\\ \hspace{1pt} \bar\nu_R & \bar e_R & \bar\nu_L & \bar e_L \\
\bar u^1_R & \bar{d}^{\,1}_R & \bar u^1_L & \bar{d}^{\,1}_L \\ \bar
u^2_R & \bar{d}^{\,2}_R & \bar u^2_L & \bar{d}^{\,2}_L \\ \bar u^3_R &
\bar{d}^{\,3}_R & \bar u^3_L & \bar{d}^{\,3}_L
\end{bmatrix}\, ,
$$
where the indices 1,2,3 correspond to the color quantum number.

The representation $\pi_F$ of $A_F$ is diagonal in generations and
$\pi_F(\lambda, q ,m)$ is given on $H_f$ i by left multiplication by
the matrix:
\begin{equation}\label{eq:8t8}
\begin{bmatrix} \left[\!
\begin{array}{c|c}
\begin{matrix} \;\lambda\; & \;0\; \\ 0 & \bar\lambda \end{matrix} &
\begin{matrix} \;0\; & \;0\; \\ 0 & 0 \end{matrix} \\ \hline
\begin{matrix} \;0\; & \;0\; \\ 0 & 0 \end{matrix} & q
\end{array} \!\right] \!\! & {0_4} \\ 0_4 & \!\!  \left[\!
\begin{array}{c|ccc} \lambda & \;0\; & \;0\; & \;0\; \\ \hline
\begin{matrix} \;0\; \\ 0 \\ 0 \end{matrix} && m
\end{array} \!\right]
\end{bmatrix}.
\end{equation} Note that $\pi_F(A_F)$ is a real $*$-algebra of
operators, and to get its complexification $\mathbb{A}_{F}$ just
replace $\bar\lambda$ by an independent $\lambda'\in \CC$, and take
$q\in M_2(\CC)$.

The grading (the chirality operator) is
$$\chi_M\otimes\chi_F,$$
where $\chi_F$ is diagonal in generations and on $H_f$ reads:
\begin{equation}\label{eq:Sgamma} \chi_F=
\begin{bmatrix} 1_2 \\ & \!\!-1_2 \\ && 0_4 \end{bmatrix} \otimes 1_4
+
\begin{bmatrix} 0_4 \\ & \!\!-1_4 \end{bmatrix} \otimes
\begin{bmatrix} 1_2 \\ & \!\!-1_2 \end{bmatrix} \;.
\end{equation}

The real conjugation is
$$J_M\otimes J_F,$$
where $J_F$ is also diagonal in generations and on $H_f$ reads:
\begin{equation}\label{eq:JF} J_F\begin{bmatrix} v_1 \\
v_2 \end{bmatrix}=\begin{bmatrix} v_2^* \\ v_1^* \end{bmatrix} \; ,
\end{equation} that satisfies $J_F^2=1$, the order 0 condition:
$$
[a,J_FbJ_F^{-1}]=0 \qquad\forall\;a,b\in A_F,
$$
and the order 1 condition:
\begin{equation}\label{eq:1storder} [[D,a],J_FbJ_F^{-1}]=0
\qquad\forall\;a,b\in A_F
\end{equation} (as in the classical case).

Finally, the Dirac operator is
$$D=\D_M\otimes\id +\chi_M\otimes D_F,$$
where $D_F$ employed by Chamseddine-Connes' reads on $H_F$
{\small\begin{multline}\label{DF} \hspace*{-5pt} D_F=
\text{\footnotesize$\left[\begin{array}{llrr|lccc} \zero & \zero &
\hspace*{-20pt}\!\!\!\!\!\!\Upsilon^*_\nu & \zero & \!\!\Upsilon^*_R &
{\zero} & {\zero}\\ \zero & \zero & \zero & \!\!\!\!\!\!\Upsilon^*_e &
\zero & {\zero} & {\zero} \\ \!\!\!\!\Upsilon_\nu\!\! & \zero & \zero
& \zero & \zero & \zero & \zero \\ \zero &
\hspace*{-5pt}\!\!\!\!\Upsilon_e & \zero & \zero & \zero & \zero &
\zero \\ \hline \!\!\!\!\Upsilon_R & \zero & \zero & \zero \!& \zero &
\!{\zero} & \zero \\ {\zero} & {\zero} & \zero & \zero &\! {\zero} &
\zero & \zero\\ {\zero} & {\zero} & \zero & \zero \!& \zero & \zero &
\zero\\ {\zero} & {\zero} & \zero & \zero \!& \zero & \zero & \zero
\end{array}\right]$}\otimes\;e_{11} +
\text{\footnotesize$\left[\begin{array}{llrr|cccc} \zero & \zero &
\hspace*{-20pt}\!\!\!\!\!\!\Upsilon^*_u & \zero & \zero & {\zero} &
{\zero} & {\zero} \\ \zero & \zero & \zero & \!\!\!\!\!\!\Upsilon^*_d
& \zero & {\zero} & {\zero} & {\zero} \\ \!\!\!\!\Upsilon_u\!\! &
\zero & \zero & \zero & \zero & \zero & \zero & \zero \\ \zero &
\!\!\!\!\Upsilon_d & \zero & \zero & \zero & \zero & \zero & \zero \\
\hline \zero & \zero & \zero & \zero \!& \zero & \!{\zero} & \zero &
\zero \\ {\zero} & {\zero} & \zero & \zero &\! {\zero} & \zero & \zero
& \zero \\ {\zero} & {\zero} & \zero & \zero \!& \zero & \zero & \zero
& \zero \\ {\zero} & {\zero} & \zero & \zero \!& \zero & \zero & \zero
& \zero
\end{array}\right]$}\otimes\;e_{11}^\perp\\ + e_{55}\otimes\;\left[\!
\text{\footnotesize$
\begin{array}{c|c}
\begin{matrix} 0\! & 0 \\ 0\! & 0 \end{matrix} &
\begin{matrix} \Upsilon^*_\nu \! & 0 \\ 0\! &
\Upsilon^*_e \end{matrix} \\ \hline \smallskip
\begin{matrix} \Upsilon_\nu\! & 0 \\ 0\! & \Upsilon_e \end{matrix}
& \begin{matrix} 0\! & 0 \\ 0\! & 0 \end{matrix}
\end{array} $} \!\right] + (e_{66}+e_{77}+e_{88})\otimes\;\left[\!
\text{\footnotesize$
\begin{array}{c|c}
\begin{matrix} 0\! & 0 \\ 0\! & 0 \end{matrix} &
\begin{matrix} \Upsilon^*_u \! & 0 \\ 0\! & \Upsilon^*_d \end{matrix}
\\ \hline \smallskip
\begin{matrix} \Upsilon_u\! & 0 \\ 0\! & \Upsilon_d \end{matrix}
& \begin{matrix} 0\! & 0 \\ 0\! & 0 \end{matrix}
\end{array} $} \!\right].
\end{multline}} 
Here the first tensor factor acts by the left matrix
multiplication and the second one by the right matrix multiplication,
$e_{jk}$ are the usual matrix units, the empty spaces stand for 0, and
$\Upsilon$'s are in Mat$(N,\CC)$ with $N$ equal to the number of
generations ($N=3$ on the current experimental basis).

Concerning the data $D_F$, $\chi_F$ and $J_F$ given above the
KO-dimension comes as 6.
\subsection{Properties of $\nu$S.M.}\label{propsnsm}

With the ingredients as listed in the previous section one gets:
\begin{itemize}
\item the group ${\mathcal G}:=\{U=uJ_FuJ_F^{-1}\,|\, u\in A_F,
det\,U=1 \}$ turns out to be isomorphic (up to a finite center) with
the gauge group $U(1)\!\times\! SU(2)\!\times\! SU(3)$ of the
S.M. (also as functions on $M$);
\item all the fundamental fermions in $H$ have the correct
S.M. charges with respect to ${\mathcal G}$\: broken to
$U(1)_{em}\times SU(3)$;
\item the 1-forms $a[D_F,b],\, a,b\in C^\infty(M, A_F)$ yield the
gauge fields $A_\mu$, $W^\pm$, $Z$, $G_\mu$ of the S.M. (from the part
$D_M$ of $D$), plus the weak doublet complex scalar Higgs field (from
the part $D_F$ of $D$).
\end{itemize} The merits of the noncommutative formulation are:
\begin{itemize}
\item it treats discrete and continuous spaces (or variables) on the
same footing;
\item both the gauge and the Higgs field arise as parts of a
connection;
\item the appearance of solely fundamental representations of
${\mathcal G}$ in the S.M. gets an explanation as the fact that they
are the only irreducible representations of simple algebras;
\item there is an elegant spectral action \, Tr$f(D/\Lambda)$, that
reproduces the bosonic part of ${\mathcal L}_{S.M.}$ as the lowest
terms of asymptotic expansion in $\Lambda$, and the matter action
$<\phi, D\phi>$ for the (Wick-rotated) fermionic part;
\item it couples in a natural way to gravity on $M$;
\item is claimed \cite{CC12,CC12a} to predict new relations among the
parameters of S.M.
\end{itemize} Some of the shortcomings still present are as in the
usual S.M.:\\ the 3 generations (families) put by hand, several free
parameters, though most of them incorporated into a single geometric
quantity: $D_F$.  Others are the unimodularity condition to be posed
on the gauge group ${\mathcal G}$ and a special treatment needed for
the two kinds of fermion doublings due to the presence of chirality
$\pm 1$ and particles/antiparticles both in $L^2(S)$ and $H_F$.

\subsection{The geometric nature of $H_f$}\label{gnnsm} The above
``almost commutative'' geometry is described by a S.T.
\begin{center} $(C^\infty(M), L^2(S),\D )\times (A_F, H_f, D_F)$,
\end{center} that is mathematically a product of the ``external''
canonical S.T. on spin manifold $M$ with the ``internal'' {\em finite}
S.T.

Few quite natural questions are in order concerning the geometric
interpretation of the internal S.T. $(A_F, H_f, D_F)$:
\begin{enumerate}[label=]
\item Does it also correspond to a (noncommutative) spin manifold?
\item Are the elements of $H_f$ ``spinors'' in some sense?
\item In particular ``Dirac spinors''?
\item Or does it correspond rather to de-Rham forms?
\item Or else?
\end{enumerate} To answer these questions, motivated by the classical
case as in Sect.\,\ref{cst}, the following definition has been
formulated for a general unital S.T. :
\begin{Def}[\!\!\!\cite{dd16}]\label{df:propM} A real spectral triple
$(A,H,D,J)$ is called \emph{spin} (and elements of $H$ are \emph{quantum Dirac
spinors}) if $H$ is a Morita equivalence $\Cl_D(A)$-$JAJ^{-1}$
bimodule (i.e. after norm-completion the algebras $\Cl_D(A)$ and
$JAJ^{-1}$ are maximal one with respect to the other).
\end{Def}

Is then the internal S.T. of $\nu$S.M. {\em spin}, like the external
one that is given by the canonical S.T. on $M$?

Building on and extending the classifications of \cite{Kra97} and
\cite{PS96} the answer in \cite{dd16} is negative. In fact therein
after a systematic search  an element 
$$X = e_{55}\otimes (1-e_{11})$$
has been found, such that $X\in \Cl_D(A)'$ but $X\notin JAJ$. A
possible way to circumvent this ``no go'' has been suggested by
employing a different grading and adding two extra non-zero matrix
elements of $D_F$, the status of which however requires a further
scrutiny (since though desirable for the correct renormalized Higgs
mass, they would have unobserved couplings to fermions).

But then, without such additions, may be the internal S.T. of
$\nu$S.M. is rather an analogue of the other natural classical
spectral triple, namely de-Rham forms?

To answer this question we have to formulate also these notions
noncommutatively using the algebraic description of the Hodge-de\,Rham
spectral triple with the grading $\chi_\Omega$ and real structure
$J_\Omega'$ as in Section \ref{hdrst}.
\begin{Def}[cf.\cite{dds17}]
\label{df:propM1} A spectral triple $(A,H,D)$ is called
\underline{complex Hodge} (and vectors in $H$ \emph{complex quantum
de-Rham forms}) if $H$ is a Morita equivalence $\Cl_D(A)-\Cl_D(A)$
bimodule (i.e. after norm-completion these algebras are maximal one
with respect to the other).\\ A complex Hodge spectral triple
$(A,H,D)$ with real structure $J$ is called \underline{Hodge} if $J$
satisfies the order 2 condition and implements the right
$\Cl_D(A)$-action.
\end{Def}

The following theorem provides the answer in the case of one
generation and thus for $\Upsilon's\in\CC$ in \eqref{DF}.

\begin{thm}[\!\!\cite{dds17}]\label{rem:CC} For the internal spectral
triple of the $\nu$S.M. with one generation the Hodge property holds
whenever $\Upsilon_x\neq 0,\;\forall\;x\in\{\nu,e,u,d\}$ and
\vspace*{-3pt}
\[\label{pred} |\Upsilon_\nu|\neq|\Upsilon_u| \qquad\text{or}\qquad
|\Upsilon_e|\neq|\Upsilon_d| \;.
\]
\end{thm}

In the rest of this section we will sketch the steps of the proof.

First by direct computation we find that the commutant of $A_F$ in
$M_8(\CC)$ is the algebra $C_F$ with elements of the
form \begin{equation}\label{eq:CF} \left[\begin{array}{ccccc}
\multicolumn{1}{c|}{\rule{-2pt}{-10pt}q_{11}} &&&
\multicolumn{1}{|c|}{\!\rule{0pt}{12pt}q_{12}} \\[2pt]
\cline{1-2}\cline{4-4} &
\multicolumn{1}{|c|}{\rule{0pt}{12pt}\,\alpha\,} \\[2pt] \cline{2-3}
&& \multicolumn{1}{|c|}{\;\beta 1_2\rule[-10pt]{0pt}{25pt}} \\
\cline{1-1}\cline{3-4} \multicolumn{1}{c|}{\rule{0pt}{5pt}q_{21}} &&&
\multicolumn{1}{|c|}{\rule{0pt}{12pt}q_{22}} \\[5pt]
\cline{1-1}\cline{4-5}
\begin{matrix} ~ \\ ~ \\ ~ \end{matrix} &&&&
\multicolumn{1}{|c}{\quad\;\delta 1_3\quad\rule[-15pt]{0pt}{0pt}}
\end{array}\right] \;,
\end{equation} where $\alpha,\beta,\delta\in\CC$, $q=(q_{ij})\in
M_2(\CC)$. Consequently the commutant of $A_F$ in
$\operatorname{End}_{\CC}(H)$ is $A_F'=C_F\otimes M_4(\CC) \simeq
M_4(\CC)^{\oplus 3}\oplus M_8(\CC)$ of complex dimension~$112$.

Next $J_FA_FJ_F\subset\operatorname{End}_{\CC}(H_F)$ consists of
elements of the form:
$$
\begin{bmatrix} 1_4\! & 0_4 \\ 0_4\! & 0_4\end{bmatrix}\otimes
\left[\!
\begin{array}{c|ccc} \lambda & 0\! & 0\! & 0\! \\ \hline
\begin{matrix} 0\!\! \\ 0\!\! \\ 0\!\! \end{matrix} && m
\end{array} \!\right] +\begin{bmatrix} 0_4\! & 0_4 \\ 0_4\! &
1_4 \end{bmatrix} \otimes\left[\!
\begin{array}{c|c}
\begin{matrix} \lambda\! & 0 \\ 0\! & \bar\lambda \end{matrix} &
\begin{matrix} 0\! & 0 \\ 0\! & 0 \end{matrix} \\ \hline
\begin{matrix} 0\! & 0 \\ 0\! & 0 \end{matrix} & q
\end{array} \!\right],
$$
where the first factors of the tensor product acts by left matrix
multiplication and the second factor by the right matrix
multiplication.

Note that $A$ and $A_{\CC}$ have the same commutant in
$\operatorname{End}_{\CC}(H_F)$.  The map $a\mapsto J_F\bar aJ_F$
gives an isomorphism $A_F\to J_FA_FJ_F$ and of their
complexifications, and also the map $x\mapsto J_F\bar xJ_F$ is an
isomorphism between $A_F'$ and $(J_FA_FJ_F)'$.

Therefore the commutant $(J_FA_FJ_F)'$ of $J_FA_FJ_F$ has elements
\begin{equation}\label{eq:lemma2} a\otimes e_{11}+
\begin{bmatrix} b \\ & c
\end{bmatrix}\otimes e_{22} +\begin{bmatrix} b \\ & d
\end{bmatrix}\otimes (e_{33}+e_{44})
\end{equation} with $a\in M_8(\CC)$, $b,c,d\in M_4(\CC)$.

Furthermore $A_F'\cap (J_FA_FJ_F)'\simeq\CC^{\oplus 10}\oplus
M_2(\CC)$.  It follows that the complex dimension of
$A_F'+(J_FA_FJ_F)'$ is $2\cdot 112-14= 210$.  The (real) subspace of
hermitian matrices has real dimenson $210$.

Now we recall that any unital complex $*$-subalgebra of
$\operatorname{End}_{\CC}(H)$, where dim$H<\infty$, is a finite direct
sum of matrix algebras: $B\simeq\bigoplus_{i=1}^sM_{m_i}(\CC)$.  The
units ${P}_i$, $1\leq i \leq s$ of $M_{m_i}(\CC)$ are orthogonal
projections and $H$ decomposes as $H\simeq\bigoplus_{i=1}^sH_i$, with
\begin{equation}\label{eq:Hi} {H}_i = {P}_i\cdot H \simeq \CC^{m_i}
\otimes \CC^{k_i} \;,
\end{equation} where $k_i$ is multiplicity of the (unique) irreducible
representation $\CC^{m_i}$ of $M_{m_i}(\CC)$ in ${H}_i$, and
$M_{m_i}(\CC)$ acts on the 1st factor of $ \CC^{m_i} \otimes
\CC^{k_i}$ by matrix product.  Then one has the following lemma:
\begin{lemma}[A]\label{lemma:11} The commutant of $B$ in
$\operatorname{End}_{\CC}(H)$ is
$B'\simeq\bigoplus\nolimits_{i=1}^sM_{k_i}(\CC)$ and the action of
$B'$ on ${H}_i \simeq \CC^{m_i} \otimes \CC^{k_i}$ is given by matrix
multiplication by $M_{k_i}(\CC)$ of the second factor in the tensor
product.
\end{lemma}

We will also need:
\begin{lemma}[B]\label{lemma:12} Let $(A,H,D,J)$ be a
finite-dimensional real spectral triple. Assume that
$B\subseteq\operatorname{End}_{\CC}(H)$ is a unital complex
$*$-algebra satisfying:
$$
\Cl_D(A)\subseteq B \qquad\text{and}\qquad B'=JBJ^{-1} \;.
$$
The following are equivalent:
\begin{enumerate}[label=(\alph*)]
\item $\Cl_D(A)' = J\Cl_D(A)J^{-1}$ (the Hodge property)\\ \item
$\Cl_D(A)'\subseteq JBJ^{-1}$\\ \item $\Cl_D(A)=B$.
\end{enumerate}
\end{lemma}
\begin{proof}[Proof of Lemma\,(B)]\leavevmode
\begin{list}{}{\leftmargin=1.0em \itemsep=3pt}
\item[(a$\Rightarrow$b)] The hypothesis $\Cl_D(A)\subseteq B$ implies
$J\Cl_D(A)J^{-1} \subseteq JBJ^{-1}$; and thus from (a) follows (b).

\item[(b$\Rightarrow$c)] $\Cl_D(A)'\subseteq JBJ^{-1}= B'$ implies $B
\subseteq\Cl_D(A)$, and so using our assumptions: $\Cl_D(A)=B$.
\item[(c$\Rightarrow$a)] If (c) holds then $B'=JBJ^{-1}$ translates to
$\Cl_D(A)'=J\Cl_D(A)J^{-1}\!.$
\end{list}
\end{proof}

Now, proceeding with the proof of Theorem, we take
\begin{equation}\label{eq:Bprime} B:=\CC\oplus M_3(\CC)\oplus
M_4(\CC)\oplus M_4(\CC)
\end{equation} \vspace{-5pt} with $(\lambda,m,a,b)\in B$ represented
on $H_F$ as
\begin{equation}\label{eq:ClDA} \left[\!
\begin{matrix} \;\lambda &0\; \\ 0 & m \end{matrix} \right] \otimes
e_{22}\otimes 1 +a\otimes e_{11}\otimes e_{11} +b\otimes e_{11}\otimes
(1-e_{11}) \;.
\end{equation} Next we:
\begin{itemize}
\item check that $\Cl_{D_F}(A_F)\subset B$;
\item check that $B$ and $J_FBJ_F^{-1}$ commute, and so
$J_FBJ_F^{-1}\subseteq B'$;
\item note that \eqref{eq:ClDA} is equivalent to the representation of
$B$ on:
$$
(\CC\otimes \CC^4) \oplus (\CC^3 \otimes \CC^4) \oplus (\CC^4 \otimes
\CC) \oplus (\CC^4 \otimes \CC^3)
$$
given by matrix multiplication on the first factors;
\item use Lemma~(A) to infer that
$$B'\simeq M_4(\CC) \oplus M_4(\CC) \oplus\CC \oplus M_3(\CC) \simeq B
$$
and so we have $J_FBJ_F^{-1}=B'$;
\item find that $\Cl_{D_F}(A_F)'\subseteq J_FBJ_F^{-1}$;
\item finish the proof by Lemma~(B).
\end{itemize}

\section{Conclusions}

The Connes-Chamseddine noncommutative formulation of the Standard
Mo\-del interprets the geometry of the S.M. as gravity on the product
of a spin manifold $M$ with a finite noncommutative `internal' space
$F$.  The multiplet of fundamental fermions (each one a Dirac spinor
on $M$) defines fields on $F$ that constitute $H_F$.

We show that the geometric nature of the latter one is not a
noncommutative analogue of Dirac spinors on $F$ (unless >2 new
parameters are introduced in the matrix $D_F$, so fields on $M$ with
physical status under scrutiny) but rather of de-Rham forms on $F$ if
the conditions \eqref{pred} are satisfied (for one generation).
Conversely, the geometric qualification of the internal spectral
triple as being Hodge constrains somewhat the parameters $\Upsilon$
occurring in the matrix $D_F$.

What happens for 3 generations of particles and so $96\times 96$
matrices?\\ It can be seen (not easily) that then the spin property
also does NOT hold, and that as adverted in \cite{BF14a}, indeed the
order 2 condition $\Cl_D(A)' \supset J\Cl_D(A)J$ holds.  Whether the
Hodge property is satisfied, or $H_F$ corresponds rather to three copies of
quantum de\,Rham forms on $F$, is currently under investigation.

\end{document}